\documentclass[10pt,twocolumn]{IEEEtran}

\usepackage{times}
\usepackage{amssymb}
\usepackage{amsmath}
\usepackage{graphicx}
\usepackage{epsfig}
\usepackage{subfigure}

\def\shadowbox{\hbox{\rule[-0.0ex]{0.1ex}{1.2ex}%
\hspace{-0.1ex}\rule[-0.0ex]{1.2ex}{0.1ex}%
\hspace{0.0ex}\rule[-0.0ex]{0.1ex}{1.2ex}\hspace{-1.3ex}%
\rule[1.15ex]{1.25ex}{0.1ex}\hspace{-0.0ex}\rule[-0.25ex]{0.3ex}{1.1ex}%
\hspace{-1.2ex}\rule[-0.25ex]{1.1ex}{0.25ex}}}
\def\qed{\ifmmode \hbox{\hfill\shadowbox}
     \else \hphantom{x}\hfill\shadowbox \fi}

\newtheorem{theorem}{Theorem}[section]
\newtheorem{lemma}[theorem]{Lemma}
\newtheorem{definition}[theorem]{Definition}
\newtheorem{proposition}[theorem]{Proposition}

\newtheorem{Signal Set}[theorem]{Signal Set}

\usepackage{amssymb}
\usepackage{amsmath}

\def\R {\mathbb{R}}
\def\Z {\mathbb{Z}}
\def\N {\mathbb{N}}
\def\C {\mathbb{C}}
\def\O {\mathcal{ O}}
\def\F {\mathcal{F}}

\def\s {\hat{\sigma}}
\def\t {\hat{\tau}}

\def\Ls {\mathbf{L}_{\sigma}}
\def\Lt {\mathbf{L}_{\tau}}

\def\p {\psi_{k,l}^{r}}
\def\pt {\psi_{k,l}^{t}}

\def\A {\mathbf{A}}
\def\sh {\hat{\sigma}}

\def\bigskip {\hspace{1in}}
\def\spot {\hspace{.3in}}

\def\psis {\psi_{s}}

\def\Rt {\mathbb{R}^{2}}

\def\r {\frac{1}{\rho}}
\def\rt {\frac{1}{\rho^{2}}}
\def\S  {\mathcal{S}}
\def\Sp  {\mathcal{S^+}}

\def\ab {\frac{\alpha}{\beta}}
\def\ba {\frac{\beta}{\alpha}}

\def\Wpp {W(\psi,\psi)}
\def\pswf{{\varphi}}
\def\projpswf{{\bf P}}

\def\range {\operatorname{range}}
\def\psis {\psi_{s}}
\def\rb {\frac{\rho}{b}}
\def\ra {\frac{\rho}{a}}

\def\ab {\frac{\alpha}{\beta}}
\def\ba {\frac{\beta}{\alpha}}
\def\psD {\frac{\pi s D}{4}}
\def\fpsD {\frac{\pi}{4s}D}

\def\I {\mathcal{I}}

\def\bound {e^{-\frac{\rho}{2}(\beta+\alpha)} }
\def\As {\mathcal{A}}
\def\Ws {\mathcal{W}}
\def\Sb {\mathbf{S}}
\def\Cb {\mathbf{C}}
\def\J {\mathcal{J}}
\def\pt {\psi^{t}}
\def\pr {\psi^{r}}
\def\CN {\mathbf{C}_{N}}

\def\LtR {L^2(\mathbb{R})}
\def\Proj {\textnormal{\bf P}}
\def\Pp {{\Proj_{\Phi}}}
\def\pf {\frac{\pi}{4}}
\def\bas {(\frac{\beta}{\alpha})^2}
\def\abs {(\frac{\alpha}{\beta})^2}

\def\HS {\textnormal{HS}}
\def\rank {\textnormal{rank}}
\allowdisplaybreaks[3]
\def\epsilon{\varepsilon}

\begin{document}

\title{Eigenvalue Estimates and Mutual Information for the Linear Time-Varying Channel}

\author{Brendan~Farrell,~\IEEEmembership{Member,~IEEE,}
        and Thomas~Strohmer
\thanks{\noindent B.~Farrell was with the Department of Mathematics,
University of California, Davis when the majority of this work 
was completed. He is now with the Lehrstuhl f\"ur Theoretische Informationstechnik, 
Technische Universit\"at M\"unchen, Arcisstr. 21, 80333 M\"unchen, Germany.
T.~Strohmer is with 
the Department of Mathematics, University of California, Davis, CA, 95616, USA.
\newline
E-mail: farrell@tum.de,\;strohmer@math.ucdavis.edu
\newline
B.~Farrell was  supported by NSF VIGRE grant DMS-0135345. 
B.~Farrell and T.~Strohmer were supported by NSF grant DMS-0511461, and T.~Strohmer 
was supported by AFOSR grant no. 5-36230.5710 and NSF grant DMS-0811169.}}

\maketitle

\begin{abstract}
We consider linear time-varying channels with additive white Gaussian
noise. For a large class of such channels we derive rigorous estimates of the 
eigenvalues of the correlation matrix of the effective channel in terms of 
the sampled time-varying transfer function and, 
thus, provide a theoretical justification for a relationship that 
has been frequently observed in the literature. We then use this
eigenvalue estimate to derive an estimate of the mutual information of the 
channel. Our approach is constructive and is based on a careful balance of
the trade-off between approximate operator diagonalization, signal
dimension loss, and accuracy of eigenvalue estimates.
\end{abstract}


\section{Introduction}

\subsection{Motivation} 

The linear, time-invariant (LTI) channel with impulse response $h$
\begin{equation}
r(t)=\int h(t-\tau)s(\tau)d\tau\label{eqLTI}
\end{equation} 
and additive white Gaussian noise with variance $\sigma^2$ 
has normalized capacity 
\begin{equation}
\frac{1}{2W}\int_{-W}^W \log\left(1+\frac{|\hat{h}(\omega)|^2}{\sigma^2}\right) d\omega
\label{capLTI}
\end{equation}
for signals band-limited to $[-W,W]$. 
This classical result is, of course, due to Shannon~\cite{Sha49},
and is probably the most fundamental result in information theory. 
We refer to~\cite{Gal68} for the mathematical steps and the 
information-theoretic details for establishing~\eqref{capLTI}.

The linear, time-variant (LTV) channel is given by 
\begin{equation}
r(t)=\int h(t,t-\tau)s(\tau)d\tau \label{eqLTV}.
\end{equation}
Motivated by Shannon's groundbreaking result, it has been a longstanding 
desire of engineers and mathematicians to derive a characterization of 
the capacity of time-varying channels in terms of the associated 
time-varying transfer function, analogous to~\eqref{capLTI}. While such a 
characterization seems still quite out of reach for the general case, 
our aim in this paper is to get one step closer to this ambitious goal.
The mathematical foundation for Shannon's famous result is the fact that
in the time-invariant case the (generalized) eigenvalues of the channel
matrix are directly related to samples of the transfer function.
Thus it is natural to ask to what extent such a relationship carries
over to the time-varying case, which is what we plan to answer in this
paper. 

For information-theoretic studies of some special cases of time-varying 
channels we refer the reader to~\cite{BPS95} and its vast list of references.
In this paper we focus on the class of time-varying 
channels whose spreading function decays at an exponential rate both
in time and frequency. This channel class is motivated by physical
properties of channel propagation and includes for instance underspread
channels~\cite{Bel63,Jak74}.

\subsection{Contributions}

A precise formulation of the results of this paper requires several 
steps of preparation. Therefore we delay the rigorous presentation
of our results to later sections, and instead give an informal
description of our contributions.

The main result of our paper shows that the eigenvalues of the 
correlation matrix 
of the effective channel can be well approximated via sampling values of 
the autocorrelation of the time-varying transfer function. We derive
rigorous bounds for the accuracy of this approximation. 
Our approach is constructive and is based on a careful balance of
the trade-off between approximate matrix diagonalization, signal
dimension loss, and accuracy of eigenvalue estimates.
While the proof of the eigenvalue estimate is quite delicate, 
this will come as no surprise to the expert in pseudodifferential operator 
theory, since characterizing the spectrum of a pseudodifferential
operator (which is essentially an operator of the form~\eqref{eqLTV}) via 
its symbol has always been a difficult task. 

We then show how this eigenvalue estimate can be used to derive an
estimate of the mutual information of these channels.
Recall that for the time-invariant 
case the mutual information (and thus in turn the capacity) is {\em precisely} 
captured by the sampled Fourier transform of the autocorrelation of the
impulse response, as the time interval is extended to infinity. 
Building on our eigenvalue estimates, we rigorously relate the mutual 
information to samples of the Fourier transform of the 
``twisted auto-convolved'' spreading function. 

\subsection{Remarks on the proof strategy} 

A few comments on the proof strategy seem in order. Two different types
of signal sets will play an important role: Weyl-Heisenberg signals
and prolate spheroidal wave functions.
The reader may wonder why we do not stick with just one of these two
types. The reason is that each of the two has some major advantages, but
also some significant limitations. Thus, by introducing both types,
Weyl-Heisenberg signals and prolate spheroidal wave functions (PSWFs), 
we can fully utilize the positive properties of each set, while mitigating 
its negative properties with the other set.

For the eigenvalue estimate we rely on a set of well localized
Weyl-Heisenberg signals whose span is close to the span of the PSWFs
in a sense that will be formalized in the proof.
While the PSWFs are optimally localized in an $L^2$-sense, their lack
of sufficient temporal decay (except for the first few PSWFs)
prohibits us from linking the eigenvalues of $\A^{\ast} \A$,
the correlation matrix of the effective channel, to
the associated time-varying transfer function.
The off-diagonal entries of the resulting matrix would have at best linear
decay, which is simply insufficient for any reasonable estimate.
On the other hand, the excellent localization properties of the 
Weyl-Heisenberg set yield an approximate diagonalization of the channel,
so that the off-diagonal entries of $\A^{\ast} \A$ decay exponentially,
which allows us to obtain a rather tight eigenvalue estimate. 

The mutual information will depend on the type and number of transmission
signals. We use a signaling set consisting of about $2TW$ mutually 
orthonormal $W$-bandlimited signals which are ``essentially localized'' 
to a time interval of length $T$. The associated signal space, rigorously 
defined in Definition~\ref{localrect},
will be denoted by $L^2(T,W,\epsilon)$. It is not difficult to construct
a {\em linear independent, well-localized} set of Weyl-Heisenberg signals.
However due to the infamous Balian-Low theorem (see Subsection~\ref{ss:wh}) 
such a set will be necessarily incomplete in $L^2(\R)$, which in turn
implies that the number of Weyl-Heisenberg signals inside $L^2(T,W,\epsilon)$
is somewhat smaller than $2TW+1$, the approximate dimension of 
$L^2(T,W,\epsilon)$. This dimension loss makes a direct estimate of the 
mutual information somewhat cumbersome. And that is where PSWFs come into play. 
We (approximately) represent $L^2(T,W,\epsilon)$ via the PSWFs, and 
then quantify the (small) dimension loss between the Weyl-Heisenberg set
and the PSWFs. Combining this estimate with our eigenvalue estimate
enables us then to estimate the mutual information in terms of
the time-varying transfer function.
 
\subsection{Connections to prior work} 

Our work is related to previous research on two aspects of time-varying 
channels. Previous authors have discussed diagonalizing the channel and 
giving the capacity in terms of singular values \cite{MG95,Dig97,BS99}, and 
other authors have focused on determining transmission signals with 
various useful properties \cite{KM98,LKS04}. Our paper is probably closest
in spirit to~\cite{DBS07}, where the authors derive estimates for the
non-coherent capacity for certain time-varying channels by carefully 
combining signal design with approximate diagonalization.

Much of the mathematical approach to time-varying channels from a 
time-frequency analysis perspective originated with 
Kozek \cite{KM98,Koz97,KM97}. 
While he addresses issues such as the composition and estimation of time-varying channel 
operators and the time-frequency localization of transmission signals, 
his focus is a WSSUS model. 
Here we work with a deterministic channel.  

The remainder of the paper is organized as follows. At the end of this
section we introduce mathematical tools and notation used throughout the paper. 
Section~\ref{MaMT} describes our setup, the channel model and the signal model.
We derive the eigenvalue estimate in Section~\ref{s:eigen} and present
the estimate of the mutual information in Section~\ref{s:mutual}.

\subsection{Mathematical tools and notation} \label{ss:tools}

Let $f$ be a function in $\LtR$.
The \emph{modulation operator} $M_{\omega}$ is defined by 
\begin{equation}
M_{\omega}f(t)=e^{2\pi i \omega\cdot t}f(t)
\end{equation}
and the \emph{translation operator} $T_x$ is defined by 
\begin{equation}
T_x f(t)=f(t-x)
\end{equation}
for all $f\in \LtR$. 
The \emph{Fourier transform} of a function $f\in L^2(\R)$ is given by
\begin{equation}
(\F f)(\omega)=\int f(t)e^{-2\pi i \omega t}dt.
\label{FT}
\end{equation}
We also write $\hat{f}$ for $\F f$. The Fourier transform of a function
in two variables is defined by extending~\eqref{FT} in 
the usual way to two dimensions. Sometimes we need to take 
the Fourier transform of a function $f(t_1,t_2)$ with respect to the first 
or the second variable only. In this case we write $\F_1 f$ or $\F_2 f$,
respectively.
When no interval is given, integration is over all of $\R$. 
For a complex-valued function $f$, we denote its complex conjugate
by $\bar{f}$.
The eigenvalues of a matrix ${\bf A}$ are denoted by $\lambda_j({\bf A})$.

The \emph{Weyl pseudodifferential operator} $\Ls$ is defined as
\begin{equation}
\Ls f(t)=\int\int \s(\omega,x)e^{-\pi ix\omega}T_{-x}M_{\omega}f(t)d\omega dx.\label{weyl}
\end{equation}
Here $\sigma$ is called the \emph{symbol} and its Fourier transform, $\sh$, is called the \emph{spreading function}. 
We can express the composition of two pseudodifferential operators $\Ls$,
$\Lt$ in terms of their symbols. There holds
$\Ls \Lt=\mathbf{L}_{\sigma\sharp\tau}$, 
where $\sigma\sharp\tau =\mathcal{F}^{-1}(\s \natural \t)$ denotes the \emph{twisted product} of $\sigma$ and $\tau$, 
and 
\begin{small}\begin{eqnarray*}
\lefteqn{(\s\natural\t) (\omega,x)=}\\
&\iint\s(\omega',x')\t(\omega-\omega',x-x')e^{-\pi i(x\omega'-\omega x')}d\omega'dx'&
\end{eqnarray*} \end{small}
is called the \emph{twisted convolution} of
$\s$ and $\t$, see~\cite{Gro01}.
This can be seen as a generalization of the composition rule of two
time-invariant operators via ordinary convolution. 

We set $\S=\overline{\sigma}\sharp\sigma$, 
which is the Fourier transform of the ``twisted autocorrelation'' of $\s$. 
Since $\S$ takes values in $\R$, $\Sp(u,v)$ is defined by 
$\Sp(u,v)=\max(S(u,v),0)$. 

\section{Channel Model and Signal Model}\label{MaMT}

We first derive an equivalent representation of the channel
model~\eqref{eqLTV}. We set $\sigma(t,\omega)=\F_2 h(t,\cdot)$. 
Several manipulations and applications of the Fourier transform yield~\cite{Gro01}
\begin{small}\begin{equation}
\int h(t,t-\tau)s(\tau)d\tau=\int\int \sh(\omega,x)M_{\omega}T_{-x}s(t)d\omega dx.
\end{equation}\end{small}
This allows us to equivalently express the linear time-varying channel as a 
pseudodifferential operator
\begin{equation}
\Ls s(t)=\int\int \s(\omega,x)e^{-\pi ix\omega}T_{-x}M_{\omega}s(t)d\omega
dx.\label{weyl1}
\end{equation}
The integral in equation~(\ref{weyl1}) has the interpretation 
that the received signal is a weighted 
sum of shifted and modulated copies of the original signal. 
Using the Weyl form allows us to express the channel as an operator that has further 
useful relationships to other forms that will be helpful in our proof. 
See~\cite{Gro01} for further background on such operators.

Our model is now given by the following steps 
and is illustrated in equations~(\ref{step1}-\ref{step5}).
First the random variable $x\in \C^N, x \sim\mathcal{NC}(0,I_N)$ 
is mapped to a set of orthonormal transmission signals $\phi_i$ as coefficients~(\ref{step1}). 
The signal passes through the channel given by $\Ls$~(\ref{step2}) 
and is corrupted by AWGN~(\ref{step3}). 
The received signal is mapped to a sequence of random variables $y$ 
by taking the inner product with the detection signals $\pr_{k}$~(\ref{step4}).

\begin{eqnarray}
\mathcal{NC}(0,I_N)\sim x &\stackrel{\Phi}{\rightarrow}&\sum_{i=1}^{N}x_{i}\phi_{i}\label{step1}\\
&\stackrel{\Ls}{\rightarrow}&\Ls \sum_{i=1}^{N}x_{i}\phi_{i}\label{step2}\\
&\stackrel{\oplus\;\textnormal{noise}}{\rightarrow}&\Ls \sum_{i=1}^{N}x_{i}\phi_{i}+n\label{step3}\\
&\stackrel{\Cb}{\rightarrow}&\hspace{-.1cm}\{ \langle \Ls \sum_{i=1}^{N}x_{i}\phi_{i}+n,\pr_{k}\rangle \}_{k\in\Z}\label{step4}\\
&=&y\label{step5}.
\end{eqnarray}
The reader will have noticed that we use a different set of signals
at the transmitter and the receiver. The mutual information between $x$ and 
$y$, $\I(x,y)$ depends on the transmission signals $\{\phi_i\}_{i=1}^N$ and 
the number of transmission signals, 
but as long as $\{\pr_k\}_{k\in \Z}$ is an orthonormal basis for $L^2(\R)$ or a 
tight frame, then $\I(x,y)$ is independent of 
the receive signals. It is clear that the transmission signals
$\{\phi_i\}_{i=1}^N$ should form a linearly independent set. 
As already briefly indicated, later the Balian-Low theorem will force us 
to select the linearly independent set of transmission signals 
from a set of functions that is also {\em incomplete} in $L^2(\R)$.
Obviously, this implies a dimension loss of the signal space which 
manifests itself in an additional error term in our main estimate of the 
mutual information.  An {\em additional} dimension loss would occur if we 
also used an incomplete signaling set at the receiver. However, at the 
receiver we are not
restricted to linearly independent signaling sets (thus the Balian-Low theorem
is no longer an obstacle) and therefore we will use a different, and
in fact overcomplete, signaling set at the receiver.

Now we introduce and discuss our requirement on the transmission signals. 
We require that they are $L^2$-localized to a time-frequency rectangle, 
which we formalize with the following definition. 

\begin{definition}\label{localrect}
We define the space $L^2(T,W,\epsilon)$ by 
\begin{small}
\begin{eqnarray*}
L^2(T,W,\epsilon)=&&\\
\lefteqn{\hspace{-2cm}\Big\{ f\in L^2(\R):\; \int_0^T|f(t)|^2dt\geq (1-\epsilon^2)\|f\|_{L^2(\R)} }\\
&\textnormal{ and }& \int_{-W}^W|\hat{f}(\omega)|^2d\omega\geq
(1-\epsilon^2)\|f\|_{L^2(\R)}\Big\}.
\end{eqnarray*}\end{small}
\end{definition}

Given the intervals $[0,T], [-W,W]$ we denote by $\{\pswf_n\}_{n=0}^{\infty}$ 
the associated PSWFs similar
to~\cite{LP62,Sle83}\footnote{The minor and trivial difference
to~\cite{LP62,Sle83} is that we consider $[0,T]$ and not $[-T,T]$.}.
Let $\projpswf$ be the orthogonal projection onto the span of
$\pswf_0,\dots,\pswf_{2TW}$. By Theorem~12 in~\cite{LP62} for every 
$f \in L^2(T,W,\epsilon)$,
\begin{equation}
\label{pswfapprox}
\| f - \projpswf f \|_2 \le 7 \epsilon \|f\|_2.
\end{equation}
In other words, $L^2(T,W,\epsilon)$ is well approximated by the first
$2TW+1$ elements of the PSWFs and 
$L^2(T,W,\epsilon)$ is essentially $(2TW+1)$-dimensional.

There are several reasons for restricting our transmission signals
to this space. Firstly, any real-world communication signal has finite
duration and (essentially) finite bandwidth. The above model is a standard
way to describe this property in a mathematically meaningful way~\cite{LP62}.
Secondly, for time-varying channels it is more insightful to
have expressions for eigenvalue estimates or mutual information for finite
time intervals (and of course finite bandwidth) than for infinite time,
as is also reflected in the papers~\cite{Gal68,MG95,DBS07}. Thus it
is useful to require some form of time-frequency localization of 
the transmission signals. We note that we could have chosen
the signal space with somewhat different localization conditions, such as
for instance using exactly time-limited signals. However, our symmetric
localization condition in Definition~\ref{localrect} lends itself to a 
somewhat shorter proof (admittedly, in spite of the overall length of our 
proof, the reader might find that using the term ``shorter'' is not 
appropriate here).

\section{Eigenvalue estimates for time-varying channels} \label{s:eigen}

\subsection{Weyl-Heisenberg systems, time-frequency localization and mutual
information} \label{ss:wh}

We assume that the reader is familiar with frame theory and refer
to~\cite{Gro01} for background.
\begin{definition}
For a given function $\phi \in \LtR$ {\em (the window function)} and 
given parameters $a,b>0$, we denote the associated \emph{Gabor system}
or \emph{Weyl-Heisenberg system} by 
$(\phi,a,b):=\{M_{b l}T_{a k}\phi\}_{k,l\in \Z}$, $a,b\in \R^{+}$.
The \emph{redundancy} of this system is $\frac{1}{ab}$. 
\emph{(}Note that $ab\leq 1$ is necessary for $(\phi,a,b)$ to be a frame
for $\LtR$ \emph{\cite{Gro01}}.\emph{)} 
\end{definition}

\begin{proposition} \label{existence}
Let $g_{s}(t)=(2s)^{-1/4}e^{-\frac{\pi}{s} t^{2}}$, and set $\psis=\Sb^{-1/2}g_{s}$, 
where $\Sb$ is the frame operator corresponding to 
$(g_{s},\frac{a}{\rho},\frac{b}{\rho})$. 
Then $(\psis,\rb,\ra)=(\psis,\rho a,\rho b)$ ($ab=1$ and $\rho>1$) 
is an orthonormal system and there exist constants $C>0$ and $0<D<1$ such that 
\[ |\psis(t)|\leq Ce^{-D\frac{\pi}{s}|t|}  \;\;\;\forall\; t\in\R\]
\[ |\widehat{\psis}(\omega)|\leq Ce^{-D\pi s|\omega|}\;\;\;\forall \;\omega\in\R.\]
\end{proposition}

\begin{proof}
A fundamental theorem due to Lyubarskii, Seip and Wallsten states that 
$(g_{s},\frac{b}{\rho},\frac{a}{\rho})$ is a frame for $L^{2}(\R)$ if and only if $\frac{ab}{\rho^{2}}<1$ \cite{Lyu92,Sei92,SW92}.
By Theorem~5.1.6 and Corollary~7.3.2 in \cite{Gro01}, 
$(\Sb^{-1/2}g_{s},\frac{b}{\rho},\frac{a}{\rho})=(\psi_{s},\frac{b}{\rho},\frac{a}{\rho})$ 
is a tight frame for $L^{2}(\R)$ with frame constant $\rho^{2}$.
Now we use the Weyl-Heisenberg biorthogonality
relations~\cite{WR90,Jan95,DLL95}, which state that if 
$\Sb_{g,\gamma}=\sum_{k,l\in\Z}\langle \;\cdot\;,M_{\beta l}T_{\alpha k}g\rangle M_{\beta l}T_{\alpha k} \gamma=I$ on $L^{2}(\R)$, 
then $\langle \gamma,M_{l/\alpha}T_{k/\beta} g\rangle =\alpha\beta \delta_{k,0}\delta_{l,0}$.
A ready consequence of this essential theorem is that 
$(\psi_{s},\frac{\rho}{a},\frac{\rho}{b})=(\psi_{s},\rho b,\rho a)$ $(ab=1)$ 
is an orthonormal set~\cite{Gro01}.
Note that $(\psi_{s},\rho b,\rho a)$ does not span $L^{2}(\R)$). 
By Theorem~5 in \cite{BJ00}, up to a factor $0<D<1$, the exponential decay of $g_{s}$ and $\hat{g}_{s}$ is 
preserved in $\psi_{s}$ and $\hat{\psi}_{s}$. 
Finally, Theorem~IV.2 in \cite{SB03} implies that if $\psi_{1}$ is the window function for the orthonormal set based 
on the initial window $g_{1}$, then  $\psi_{s}$ is the corresponding window function for $g_{s}$. 
\end{proof}

Let $\psi_s,a$ and $b$ be as in the previous proposition. 
We construct our signals by setting $a=\ba$, $b=\ab$ and $s=(\ab)^2$. 
The signals are then defined by:
\begin{itemize}
\item[D1)] $\psi^{t}=\psi_{(\ab)^2}$
\item[D2)] $\psi^{r}=\frac{1}{\rho}\psi_{(\ab)^2}$
\item[D3)] $\p =M_{\r bl}T_{\r \ba k} \psi^{r}$
\item[D4)] $\psi_{k,l}^{t}= M_{\rho \ab l}T_{\rho \ba k} \psi^{t}$
\end{itemize}
Here $t$ stands for ``transmit'' and $r$ stands for ``receive''.

\begin{definition}\label{exploc}
A function $f\in \LtR$ is \emph{exponentially localized} 
to the region $[0,T]\times[-W,W]$  
if there exist constants $c_1,C_1,c_2$ and $C_2$ such that 
\begin{equation}
|f(t)|\leq C_1 e^{-c_1|t|}\;\textnormal{ and }\; |\hat{f}(\omega)|\leq C_2e^{-c_2|\omega|}
\end{equation}
for all $t\notin [0,T]$ and all $\omega\notin [-W,W]$. 
\end{definition}

The Balian-Low theorem~\cite{Gro01} precludes the existence of an orthonormal 
Weyl-Heisenberg basis $(\phi,a,b)$ for $\LtR$ with well-localized window
function. In particular, $\phi$ and $\hat{\phi}$ could never have
exponential decay. On the other hand (as for instance
Proposition~\eqref{existence} shows) it is not difficult to construct an 
orthonormal system that is incomplete in $L^2(\R)$ or an overcomplete system $(\phi,a,b)$ with 
a $\phi$ that is exponentially well localized in time and frequency. Thus, 
the Balian-Low theorem is the reason why we use a 
signaling set at the transmitter drawn from an incomplete system for $L^2(\R)$ 
(implying $\rho>1$) and an overcomplete signaling set at the receiver.

While mutual information is not the main topic of this section, we take
the opportunity to address a non-trivial aspect associated with
mutual information that arises from using a tight frame instead of
an orthonormal basis as receive functions.
If we used an orthonormal basis at the receiver, then the noise covariance matrix, 
$\Cb_N\Cb_N^\ast$ in the proof below, 
would be a multiple of the identity, and this proposition would 
be simple and standard. 
Using a unit-norm tight frame rather than an orthonormal basis does not change 
the eigenvalues, 
but it does make the property addressed in the proposition below more delicate. 
The exponential localization at the receiver and the $L^2(T,W,\epsilon)$-property at the transmitter, 
however, deliver the necessary approximations 
for this proposition to hold.

\begin{proposition}\label{equiv}
Let $\{\phi_{kl}\}_{(k,l)\in\J}$, $|\J|<\infty$, be orthonormal transmission signals 
contained in $L^2(T,W,\epsilon)$, 
and let $\{ \pr_{kl}\}_{k,l\in\Z}$ be a tight frame of exponentially localized receiver signals 
(with frame bound $1$).
Let $x\sim\mathcal{NC}(0,I_{|\J|})$ and 
$$y_{kl}=\langle \Ls
\sum_{k'l'\in\J}x_{k'l'}\phi_{k'l'}+n,\pr_{kl}\rangle, 
\quad \text{for $k,l\in\Z$},$$ 
where $n(t)$ is AWGN of variance $\eta^{2}$. 
Denote
\begin{equation}
\A_{klk'l'}=\langle \Ls \phi_{k'l'},\pr_{kl}\rangle.
\label{Amatrix}
\end{equation}
Then 
\begin{equation}
\I(x;y)=\sum_{i=1}^{|\J|}\log\left(1+\frac{\lambda_{i}(\A\A^{\ast})}{\eta^{2}}\right).\label{eigsum}
\end{equation}
\end{proposition}

\begin{proof} 
Let $\Phi:L^2(\R)\rightarrow L^2(\R)$ be the orthogonal projection onto $\textnormal{span}\{\phi_{kl}\}_{(k,l)\in\J}$, 
and let  $\Cb:L^{2}(\R)\rightarrow l^{2}(\Z^{2})$ and  
$\CN:L^{2}(\R)\rightarrow \C^{(2N+1)^{2}}$ be the coefficient operators given by 
$\Cb f=\{\langle f, \pr_{kl}\rangle \}_{k,l\in \Z}$ and
$\CN f=\{\langle f, \pr_{kl}\rangle \}_{|k|,|l|\leq N}$ for $N\in\N$. 
The mutual information $\I(x;y)$ is 
\begin{eqnarray*}
\lefteqn{\I(x;y)}\\
&=&\lim_{N\rightarrow\infty}\{\log\det(\CN \Ls   \Phi\Ls^{\ast}\CN^{\ast}+
\eta^{2}\CN\CN^{\ast})\\
&&\hspace{2cm}-\log\det(\eta^{2}\CN\CN^{\ast})\}.
\end{eqnarray*}

Assume $ \Phi \Ls\Ls^{\ast} \Phi$ has rank $k$, 
and arrange all eigenvalues in non-increasing order.
We must show that 
\begin{eqnarray*}
\lim_{N\rightarrow\infty}\lambda_{i}(\CN \Ls
\Phi\Ls^{\ast}\CN^{\ast}+\CN\CN^{\ast})
\lambda_{i}(\Cb \Ls \Phi\Ls^{}\Cb^{\ast})+1
\end{eqnarray*}
for $i=1,...,k$.
Note that $\Cb \Ls   \Phi\Ls^{\ast}\Cb^{\ast}$ and $ \Phi\Ls\Ls^{\ast} \Phi$ 
have the same nonzero eigenvalues. 

Since $\s$ decays exponentially in both variables and each $\phi_{k'l'}\in L^2(T,W,\epsilon)$, 
using the Cauchy-Schwartz inequality shows that each $\Ls\phi_{k,l}$ 
is exponentially localized  a time-frequency rectangle slightly larger than $[0,T]\times [-W,W]$.  
Thus the range of $\Ls \Phi$ is exponentially localized in time and frequency, 
and so any eigenvectors of  $ \Ls  \Phi\Ls^{\ast}$ corresponding to nonzero 
eigenvalues, since they belong to the range of $\Ls \Phi$,
are similarly exponentially localized, 
which holds as well for $\CN \Ls  \Phi\Ls^{\ast}\CN^{\ast}$ for all $N$. 
In particular, for all $f$ in the range of $\Ls \Phi$, there exist positive
constants $c,C$ such that 
\[ \|\CN^{\ast}\CN f-f\|_{L^{2}(\R)}\leq Ce^{-cN}.\]
Let $u^{(N)}_{i}$ be an eigenvector of $\CN \Ls  \Phi\Ls^{\ast}\CN^{\ast}$ 
corresponding to the nonzero eigenvalue $\lambda_{i}$. 
Then $u^{(N)}_{i}=\CN f^{(N)}_{i}$ for some $f^{(N)}_{i}$ in 
the range of $\Ls \Phi$. 
Now,
\begin{eqnarray}
\lefteqn{\lim_{N\rightarrow\infty}(\CN f^{(N)}_{i})^{\ast}\CN\CN^{\ast}(\CN
f^{(N)}_{i})}\\
&=&\lim_{N\rightarrow\infty}\langle
f^{(N)}_{i},\CN\CN^{\ast}\CN\CN^{\ast} f^{(N)}_{i}\rangle\label{limitone}\\
&=&\lim_{N\rightarrow\infty}\langle f^{(N)}_{i},f^{(N)}_{i}\rangle\label{limittwo}\\
&=&1.\nonumber
\end{eqnarray}
The convergence in lines~(\ref{limitone}) and~(\ref{limittwo}) is exponential. 
While exponential convergence is not necessary, without sufficient localization of all the functions involved, 
convergence at all does not hold a priori for~(\ref{limitone}) and~(\ref{limittwo}). 
For $i=1,...,k$, 
\begin{eqnarray*}
\lefteqn{\lim_{N\rightarrow \infty} \lambda_{i}(\CN \Ls
\Phi\Ls^{\ast}\CN^{\ast}+\eta^{2}\CN\CN^{\ast})}\\
&=&\lim_{N\rightarrow\infty}\lambda_{i}(\CN \Ls  \Phi\Ls^{\ast}\CN^{\ast})+ \eta^{2}
\end{eqnarray*}
The remaining eigenvectors of $\CN\CN^{\ast}$ are in the kernel of 
$\CN \Ls  \Phi\Ls^{\ast}\CN^{\ast}$. 
Thus, 
\begin{eqnarray}
\lefteqn{\lim_{N\rightarrow\infty}\Big\{\sum_{i=1}^{(2N+1)^{2}}
\log(\lambda_{i}(\CN \Ls  \Phi\Ls^{\ast}\CN^{\ast}+ \eta^{2} \CN\CN^{\ast}))}\nonumber\\
&&\hspace{2cm}-\sum_{i=1}^{(2N+1)^{2}}\log(\lambda_{i}( \eta^{2}\CN\CN^{\ast}))\Big\}  \nonumber\\
&=&\lim_{N\rightarrow\infty}\sum_{i=1}^{k}\log(\lambda_{i}(\CN \Ls  \Phi\Ls^{\ast}\CN^{\ast})+\eta^{2})\nonumber\\
&&\hspace{2cm}-\lim_{N\rightarrow\infty} \sum_{i=1}^{k}\log(\lambda_{i}(\eta^{2}\CN\CN^{\ast}))\label{fewer}\\
&=&\lim_{N\rightarrow\infty}\sum_{i=1}^{k}\log(1+\frac{\lambda_{i}(\CN \Ls  \Phi\Ls^{\ast}\CN^{\ast})}{\eta^{2}})\label{allones}\\
&=&\sum_{i=1}^{k}\log(1+\frac{\lambda_{i}( \Phi^{\ast}\Ls^{\ast} \Ls  \Phi^{\ast})}{\eta^{2}})\nonumber\\
&=&\sum_{i=1}^{k}\log(1+\frac{\lambda_{i}(\A^{\ast}\A)}{\eta^{2}})\nonumber\\
&=&\sum_{i=1}^{|\J|}\log(1+\frac{\lambda_{i}(\A^{\ast}\A)}{\eta^{2}}),\nonumber
\end{eqnarray}
where lines~(\ref{fewer}) and~(\ref{allones}) 
are consequences of the first half of the proof.
\end{proof}

\begin{subsection}{Eigenvalue Estimates}\label{estimates}

We are ready to give a rigorous formulation of our main result, which
states that the eigenvalues of the correlation matrix
$\A^{\ast}\A$ can be well approximated by samples of ${\cal S}$, the
twisted autocorrelation of the time-varying transfer function.

\begin{theorem}[Eigenvalue estimate]\label{mainthm1}
Assume the same setup as in Proposition~\ref{equiv}.
Furthermore, suppose that
\begin{equation}
\label{spread_decay0}
 |\sh(\omega,x)|\leq Ce^{-\beta|\omega|-\alpha|x|}.
\end{equation}
Let $\S=\overline{\sigma}\sharp\sigma$.
Then for $j=1,...,|\J|$, there exists an index pair $(k,l)$ such that
\begin{small}\begin{eqnarray}
\left|\lambda_j(\A^\ast \A ) - \S(\rho\ba k,\rho\ab l) \right|
&\leq& \O\left(\bound +  \frac{1}{(\alpha\beta D)^{2}} \right).\nonumber\\
&& \hspace{1cm}\label{eigenvalueestimate}
\end{eqnarray}\end{small}
\end{theorem}

{\em Remark:} Our decay condition~\eqref{spread_decay0} on the 
spreading function comprises the standard conditions of
exponential decay of delay spread and compact support of the Doppler
spread~\cite{Jak74}. 
Moreover, we could have imposed an {\em underspread} condition on
the spreading function, see~\cite{MH98} for various notions of underspread 
channels. It is not hard to see that condition~\eqref{spread_decay} 
includes (or can be
easily adapted to) several forms of underspread channels.
This would result in somewhat different constants in the error 
estimate at the cost of a slightly
longer proof, but the essence of the theorem would remain the same.
Furthermore, one can replace the exponential decay condition
by some form of (practically less justified) polynomial decay 
and show that the error term in~\eqref{eigenvalueestimate} would then decrease
at a corresponding polynomial rate.

To prove Theorem~\ref{mainthm1} we cannot use PSWFs, but instead 
introduce exponentially localized signals. The reason is that the PSWFs decay 
linearly~\cite{Sle83} and, thus, do not permit the bounds obtained in the 
main two lemmas of this section. This is heuristically explained by the 
fact that the PSWFs are the approximate eigenfunctions of the operator that 
restricts in time and frequency, which is a much different operator than 
a time-varying channel, for which the exponentially localized signals are 
approximate eigenfunctions. This is seen formally in the off-diagonal decay 
in the matrix $\A$ in Proposition~\ref{diag} below. However, since both sets 
of signals are localized, the spaces that they span are close, which is 
a point that we formalize later in the proof of Theorem~\ref{mainthm2}. 
Thus, the general idea is the standard linear algebra approach of 
working with the same space, but switching to a basis that allows for 
approximate diagonalization. 

We first need an auxiliary result.
\begin{lemma}\label{WA} 
For $f,g\in L^2(\R)$, let $\mathcal{W}(f,g)$ and $\mathcal{A}(f,g)$ denote their cross-ambiguity and cross-Wigner distributions~\cite{Gro01}. 
If $|\psi(x)|\leq Ce^{-c_{1}|x|}$ and $|\hat{\psi}(\omega)|\leq Ce^{-c_{2}|\omega|}$ for $c_{1},c_{2}>0$, then 
\begin{equation*}
|\mathcal{W}(\psi,\psi)(x,\omega)|\leq C^{2}e^{-\frac{1}{4}(c_{1}|x|+c_{2}|\omega|)}
\end{equation*}
and
\begin{equation*}
|\mathcal{A}(\psi,\psi)(x,\omega)|\leq C^{2}e^{-\frac{1}{4}(c_{1}|x|+c_{2}|\omega|)}.
\end{equation*}
\end{lemma}
\begin{proof}
The proof is contained in the proof of Theorem~2.4 in \cite{Str01}, 
when one views both distributions as short-time Fourier transforms, 
as explained in $\cite{Gro01}$.
\end{proof}

A key ingredient in our proof of Theorem~\ref{mainthm1} is the following 
lemma, which shows that the 
entries of the matrix ${\bf A}$ defined in~\eqref{Amatrix} decay exponentially 
fast as we move away from the main diagonal. The approximate diagonalization 
of ${\bf A}$ via a properly designed Weyl-Heisenberg systems is well known in a 
qualitative sense~\cite{KM97,Str04,DBS07}. What is new in the following 
lemma is that we give a precise quantitative formulation of this statement. 
This quantitative version is important in the subsequent steps, where
it will give rise to explicit and rigorous bounds on the approximation 
of the eigenvalues of ${\bf A^{\ast}A}$ by samples of the twisted
autocorrelation ${\cal S}$ of the time-varying transfer function.

\begin{lemma}\label{diag}
Assume that  
\begin{equation}
 |\sh(\omega,x)|\leq Ce^{-\beta|\omega|-\alpha|x|},
\end{equation}
that the signals are given according to properties $D1-D4$ above and 
that
\begin{equation}
\A_{klk'l'}=\langle \Ls \pt_{k'l'},\pr_{kl}\rangle.
\label{Amatrix}
\end{equation}
Then
\begin{eqnarray*}
 |\A_{klk'l'}|&\leq&  C(e^{-\alpha \rho |\rt l-l'|}+e^{-\pf D \abs\rho|\rt l-l'|})\\
&& \times(e^{-\beta\rho |\rt k-k'|}+e^{-\pf D\bas \rho|\rt k-k'|}).
\end{eqnarray*}
\end{lemma}

\begin{proof}
The following two essential identities hold for pseudodifferential
operators, cf.~\cite{Gro01}:
\begin{equation}
 \langle \Ls f,g\rangle=\langle \sigma, \Ws(g,f)\rangle\label{weylwigner}
\end{equation}
\begin{equation}
 \left|\langle \Ls T_{u}M_{\eta}f,T_{v}M_{\gamma}g\rangle|=|(\sh *\As(f,g))(u-v,\eta-\gamma)\right|.\label{weylamb}
\end{equation}
The system is given by $\psi_{(\ab)^2}  =\Sb^{-1/2}g_{(\ab)^2}$, where $g_{(\ab)^2}(t)=(2(\ab)^2)^{-1/4}e^{-\pi(\ba)^2 t^{2}}$, and by Proposition (\ref{existence})
\begin{equation}
 |\psi_{(\ab)^2}(t)|\leq Ce^{-\pi(\ba)^2 D|t|}
\end{equation}
\begin{equation}
 |\hat{\psi}_{(\ab)^2}(\omega)|\leq Ce^{-\pi \abs D|\omega|}.
\end{equation}
Lemma~\ref{WA} implies 
\begin{equation}
|\As(\psi_{(\ab)^2},\psi_{(\ab)^2})(x,\omega)|\leq Ce^{-\pf D\abs|x|-\pf D\bas |\omega|}.
\end{equation}
\begin{eqnarray*}
\lefteqn{|\A_{k,l,k',l'}|}\\
&=&|\langle \Ls \pt_{k'l'},\pr_{kl} \rangle|\\
&=&|\langle \Ls M_{\rho \ab l'}T_{\rho \ba k'} \psi,M_{\r \ab l}T_{\r \ba k} \psi\rangle|\\
&=&|(\sh *\As(\psi,\psi)) (\ba (\r k-\rho k'),\ab (\r l-\rho l'))|\\
&=&| \iint \sh(\omega,x)\As(\psi,\psi)\\
&&\spot (\ba (\r k-\rho k')-x,\ab (\r l-\rho l')-\omega)dxd\omega | \\
&\leq & C\iint e^{-\beta|\omega|-\alpha|x|}\\
&&  e^{-\pf \bas D|\ba (\r k-\rho k')-x|-\pf \abs D|\ab (\r l-\rho l')-\omega|}d\omega dx\\
&=&C\int  e^{-\beta|\omega|-\pf \bas D|\ab (\r l-\rho l')-\omega|}d\omega\\
&&\times \int e^{-\alpha|x|-\pf \abs D|\ba (\r k-\rho k')-x|}dx\\
&\leq & C(e^{-\alpha  \rho |\rt l-l'|}+e^{-\pf \ab D \rho|\rt l-l'|})\\
&& \times(e^{-\beta \rho |\rt k-k'|}+e^{-\pf\ba D \rho|\rt k-k'|}),
\end{eqnarray*}
where we have used the bound:
\begin{eqnarray*}
\int e^{-c_{1}|y|}e^{-c_{2}|X-y|}dy&\leq & C(e^{-c_{1}|X|}+e^{-c_{2}|X|}).\label{littlebound}
\end{eqnarray*}
\end{proof}

The following lemma shows that the eigenvalues of ${\bf A^{\ast} A}$
are well approximated by its diagonal entries.

\begin{lemma}\label{one}
Assume again the hypotheses of Proposition~\ref{diag}. 
Then for $j=1,...,|\J|$, there exists an index pair $(k,l)$ such that 
\begin{equation}
\left|\lambda_j(\A^\ast \A )-(\A^{\ast}\A)_{klkl})\right|\leq \O(\bound). 
\end{equation}
\end{lemma}

\begin{proof}

\begin{eqnarray*}
(\A^{\ast}\A)_{klk'l'}
&=&\sum_{j,j'\in\Z}\overline{\A_{jj'kl}}\A_{jj'k'l'}\\
&=&\sum_{j,j'\in\Z}\overline{\langle\Ls \pt_{kl},\pr_{jj'}\rangle}\langle \Ls \pt_{k'l'},\pr_{jj'}\rangle\\
&=&\sum_{j,j'\in\Z}\overline{\langle\Ls \pt_{kl},\pr_{jj'}\rangle\langle \pr_{jj'},\Ls \pt_{k'l'}\rangle}\\
&=&\langle\Ls\pt_{k'l'},\Ls\pt_{kl}\rangle\\
&=&\langle \mathbf{L}_{\S}\pt_{k'l'},\pt_{kl}\rangle,
\end{eqnarray*}
where $\S=\overline{\sigma}\sharp\sigma$ was defined in Section~\ref{MaMT}. 
Using the estimate from the proof of Lemma~\ref{diag}, 
we have that $|\hat{\S}(\omega,x)|\leq \frac{C}{\alpha\beta}e^{-\beta|\omega|-\alpha|x|}$.
Using the identity in equation~(\ref{weylamb}),
\begin{eqnarray*}
\lefteqn{|\langle \mathbf{L}_{\S}\pt_{k'l'},\pt_{kl}\rangle|}\\
&=&|(\hat{\S}*\mathcal{A}(\pt,\pt))( \rho \ba(k'-k),\rho \ab (l'-l))|\\
&\leq & C(e^{-\alpha  \rho | l-l'|}+e^{-\pf \ab D \rho| l-l'|})\\
&& \times(e^{-\beta \rho | k-k'|}+e^{-\pf\ba D \rho| k-k'|}).
\end{eqnarray*}
Next
\begin{eqnarray}
\lefteqn{\sum_{\stackrel{k=-K,...,K,k\neq k'}{ l=-L,...,L,l\neq l'}}|(\A^{\ast}\A)_{klk'l'}|}\nonumber \\
&\leq& C \sum_{k=-K,..,K,k\neq k'}( e^{-\beta\rho  |k-k'|}+ e^{-\pf D \ba\rho  |k-k'|})  \nonumber  \\
&&\times \sum_{l=-L,..,L,l\neq l'}( e^{-\alpha\rho  |l-l'|}+ e^{-\pf D \ab\rho  |l-l'|})    \nonumber\\
&=&\O \left( \left(\frac{e^{-\beta\rho }}{1-e^{-\beta \rho }}+\frac{e^{-\pf D\ba\rho }}{1-e^{-\pf D \ba\rho }}\right)\right.\nonumber\\
&&\spot \times\left. \left(\frac{e^{-\alpha \rho }}{1-e^{-\alpha\rho }}+\frac{e^{-\pf D\ab\rho }}{1-e^{-\pf D \ab\rho }}\right)     \right).\label{est}\\
&=&\O\left( \frac{e^{-\rho (\beta+\alpha)}}{ (1-e^{-\pf D\ba\rho})( 1-e^{-\pf D\ab\rho })}\right)\nonumber\\
&=&\O\left( e^{-\frac{\rho}{2}(\beta+\alpha)}\right). \label{next}
\end{eqnarray}
We now have an estimate on the off-diagonal sums of the matrix 
$\A^{\ast}\A$ and may apply the Gershgorin disc theorem to obtain the claim.
\end{proof}

Having established that the spectrum of ${\bf A^{\ast} A}$ is very close
to its diagonal entries, we next show that in turn
the diagonal of ${\bf A^{\ast}A}$ is well approximated by
the samples of the associated twisted autocorrelation ${\cal S}$.

\begin{lemma}\label{two}
Assume again the hypotheses of Proposition~\ref{equiv} 
and that 
\[ |\sh(\omega,x)|\leq Ce^{-\beta|\omega|-\alpha|x|}.\]
Let $\S=\overline{\sigma}\sharp\sigma$. 
Then 
\begin{equation*}
| (\A^{\ast}\A)_{klkl}- \S(\rho\ba k,\rho\ab l)| = \O\left( \frac{1}{(\alpha\beta D)^{2}} \right).
\end{equation*}
\end{lemma}
\begin{proof}
We first look at $(\A^{\ast}\A)_{klkl}$. 
The diagonal entries of $\A^{\ast}\A$ are 
\begin{eqnarray}
(\A^{\ast}\A)_{klkl}&=&\sum_{k'l'\in\Z^{2}}|\langle\Ls\pt_{kl},\psi^{r}_{k'l'}\rangle|^{2}\label{diagnorm}\\
&=&  \|\Ls \pt_{kl}\|^{2}_{2},
\end{eqnarray}
since $(\psi,\r a,\r b)$ is a tight Weyl-Heisenberg frame (Proposition (\ref{existence})).
\begin{eqnarray}
\lefteqn{\|\Ls \pt\|^{2}_{2}}\\
&=&\langle \Ls \pt_{kl},\Ls \pt_{kl}\rangle\\
&=& \langle \overline{\sigma}\sharp\sigma,\Ws(\pt_{kl},\pt_{kl})\rangle\\
&=&\int_{\R^{2}}\S(x,\omega)\Ws(\pt,\pt)(x-\rho\ba k,\omega-\rho\ab l)d\omega dx\nonumber
\end{eqnarray}

Setting $\S'=\partial_x\partial_\omega \S$, by the Riemann-Lebesgue Lemma, 
\begin{eqnarray*}
\|\S'\|_{\infty}&\leq& \iint |\hat{S}(\omega,x)|d\omega dx\\
&\leq& C\iint \frac{1}{\alpha\beta}e^{-\frac{\beta}{2}|\omega|-\frac{\alpha}{2}|x|}d\omega dx\\
&=&\frac{C}{(\alpha\beta)^{2}}.
\end{eqnarray*}

We use Lemma~\ref{WA} and the fact that $\iint \Wpp(\omega,x)d\omega
dx=\|\psi\|_{2}^{2}=1$, cf.~\cite{Gro01}.
\begin{eqnarray*}
\lefteqn{|\|\Ls \pt\|^{2}_{2}-\S(\rho\ba k,\rho\ab l)|}\\
&=&|\int_{\R^{2}}\S(x,\omega)\Ws(\psi,\psi)(x-\rho\ba k,\omega-\rho\ab l)d\omega dx\\
&&\hspace{4cm} -\S(\rho\ab l,\rho\ba k)|\\
&=& |\int_{\R^{2}}\S(x+\rho\ba k,\omega+\rho\ab l)\Ws(\psi,\psi)(x,\omega)d\omega dx\\
&& \hspace{4cm} -\S(\rho\ab l,\rho\ba k)|\\
&=& |\int_{\R^{2}}[\S(x+\rho\ba k,\omega+\rho\ab l)-\S(\rho\ab l,\rho\ba k)]  \\
&& \hspace{3cm}  \Ws(\psi,\psi)(x,\omega)d\omega dx)|\\
&\leq &\|S'\|_{\infty}\int_{\Rt}(|x|+|\omega|)\Ws(\psi,\psi)(x,\omega)|d\omega dx\\
&\leq&C \frac{1}{(\alpha\beta)^{2}}\int_{\Rt}(|x|+|\omega|)e^{-\psD |x|-\fpsD|\omega|}d\omega dx\\
&=&C \frac{1}{(\alpha\beta D)^{2}}
\end{eqnarray*}
These two bounds prove the lemma.
\end{proof}

\begin{proof}[Proof of Theorem~\ref{mainthm1}]
The estimate~\eqref{eigenvalueestimate} follows now readily by applying 
the triangle inequality to the left-hand-side of~\eqref{eigenvalueestimate},
and then using Lemma~\ref{one} and Lemma~\ref{two}.
\end{proof}

{\em Remark:} In the proof of this theorem we rely on using
Weyl-Heisenberg systems.
Instead we could have resorted to orthonormal Wilson bases~\cite{Gro01},
which do not suffer from the Balian-Low Theorem.
However it would have resulted in a less elegant relationship 
between eigenvalues and samples of ${\cal S}$. 
In particular, equations~\eqref{weylwigner} and~\eqref{weylamb} would have to
be replaced by more complicated expressions.

\section{From Estimating Eigenvalues to Estimating Mutual Information}
\label{s:mutual}

For the time-invariant case, the mutual information is precisely captured
by samples of the Fourier transform of the autocorrelation of the 
impulse response when one allows $T\rightarrow \infty$. At the core of
this relationship is the fact that the (generalized) eigenvalues of
the channel are directly linked to samples of the transfer function.
It turns out that for our class of time-varying channels 
a similar connection is true in an approximate sense. Using the eigenvalue
estimate from the previous section we will show that one can obtain
an estimate of the mutual information via samples of the Fourier transform 
of the ``twisted auto-convolved'' spreading function. 
This is the contents of the following theorem.

\begin{theorem}[Mutual information estimate]\label{mainthm2}
Assume that the spreading function $\sh$ in the system model satisfies 
\begin{equation}
|\sh(\omega,x)|\leq Ce^{-\beta|\omega|-\alpha|x|},
\label{spread_decay}
\end{equation}
and the AWGN $n(t)$ has variance $\eta^2$. 
Let $\Phi_{T,W}=\{\phi_k\}_{k=1}^N$ be a set of 
orthonormal functions contained 
in $L^2(T,W,\epsilon)$, where $N=(1-\delta)(2TW+1)$ for some $0\leq \delta<1$. 
Let $\I_{\Phi_{T,W}}(x,y)$  denote the resulting mutual information of the system given in lines~(\ref{step1}-\ref{step5}).  
Then there exist constants $0<D$, $1<\rho$ and  
small constants $0\leq \delta_1,\delta_2$ such that 
\begin{small}\begin{eqnarray}
\lefteqn{\left|\I_{\Phi_{T,W}}(x,y)-\sum_{k=0,l=-L}^{K,L}\log\left(1+\frac{\S^{+}(\rho\ba k,\rho \ab l)}{\eta^{2}}\right)\right|}\\
&\leq&  (2TW+1)\Bigg(\log\Big(1+\O\Big( \bound+\frac{1}{(\alpha\beta D)^{2}}\Big)\Big)\label{firsterror}\\
&&\hspace{-.5cm}+\log\Big(1+ \Big(\frac{ 14\epsilon }{\eta^2}+\frac{ (14\epsilon+\delta) }{\eta^2}\Big)\|\S\|_{L^\infty(\R)}  \Big)\label{seconderror}\\
&&\hspace{-.5cm}+\log\Big(1+ \Big(\frac{ 14\epsilon}{\eta^2}+\frac{1-\frac{(1-49\epsilon^2)}{\rho^2}+\frac{\delta_1}{\rho}\ba +\frac{\delta_2}{\rho}\ab }{\eta^2}\Big) \|\S\|_{L^\infty(\R)} \Big)\Bigg)\nonumber\\
&&\label{seconderrorb}
\end{eqnarray}\end{small}
where $K=\frac{T}{\rho}{\ab}-\delta_1$ and $L= \frac{2W}{\rho}\ba-\delta_2$. 
The parameters $D$ and $\rho$ have the relationship 
that $D\rightarrow 0$ as $\rho\rightarrow 1$ and $\rho\rightarrow \infty$ as $D\rightarrow 1$. 
The numbers $\delta_1$ and $\delta_2$  depend on the parameters $\alpha,\beta$ and $\epsilon$, 
but remain small as $T$ and $W$ increase. 
\end{theorem}

Before we proceed to the proof of this theorem, it seems prudent to 
comment on the statement of this theorem and the various elements that
come into play here.

{\em Remark 1:} 
In a nutshell our theorem shows that
$$\I_{\Phi_{T,W}}(x,y) \approx \sum_{k=0,l=-L}^{K,L}
\log\Big(1+\frac{\S^{+}(\rho\ba k,\rho \ab l)}{\eta^{2}}\Big),$$
and quantifies rigorously in which sense this approximation is true.
The error due to estimating the mutual information from the 
samples is given in~\eqref{firsterror} and is the conceptually more important 
one for this paper. 
The error in~\eqref{seconderror} results from the transition from the system 
$\Phi_{T,W}$ in $L^2(T,W,\epsilon)$ to the PSWFs, and the 
error~\eqref{seconderrorb} is due to the fact that the number of the 
constructed Weyl-Heisenberg signals used is less than the number of 
PSWFs corresponding to the time-frequency region.

{\em Remark 2:} The factor $\rho$ is necessary for our construction 
and is greater than $1$, see Proposition~\ref{existence} and the subsequent
discussion.
While taking $\rho$ very close to $1$ would make the error in 
equation~(\ref{seconderrorb}) very small, it would increase the error in 
equation~(\ref{firsterror}). 
We can, however, take $\rho$ to be fairly close to $1$, 
such as $\rho=5/4$.
This issue of the trade-off between time-frequency localization and loss of 
dimensions in signal space has also been pointed out in~\cite{DBS07}.

We need the following lemma for the proof of Theorem~\ref{mainthm2}. 

\begin{lemma}\label{difference}
Let $\S=\overline{\sigma}\sharp\sigma$ and $S^{+}(x,\omega)=(S(x,\omega))^{+}$.
Then 
\begin{eqnarray*}
\lefteqn{\left|\log(1+\lambda_{k,l}(\A^{\ast}\A))-\log(1+\S^{+}(\rho\ba k,\rho\ab l))\right|}\\
& =& \log\left( 1+\O\left(\bound+\frac{1}{(\alpha\beta D)^{2}} \right) \right)
\end{eqnarray*}
\end{lemma}
\begin{proof}
Using Lemmas~\ref{one} and~\ref{two},
\begin{eqnarray*}
\lefteqn{\left|\log(1+\lambda_{k,l}(\A^{\ast}\A))-\log(1+\S^{+}(\rho\ba k,\rho\ab l))\right|} \\
&\leq& \left|\log(1+\lambda_{k,l}(\A^{\ast}\A)) -\log(1+(\A^{\ast}\A)_{klkl})\right|\\
&&+\left| \log(1+(\A^{\ast}\A)_{klkl}) -\log(1+\S^{+}(\rho\ba k,\rho\ab l))\right|\\
&=& \log\left( 1+\O\left(\bound + \frac{1}{(\alpha\beta D)^{2}} \right) \right)
\end{eqnarray*}
\end{proof}
\end{subsection}

\begin{proof}[Proof of Theorem~\ref{mainthm2}]
Let $\Proj $ denote the projection of $L^2(\R)$ onto the span 
of the $2TW+1$ PSWFs 
corresponding to $[0,T]\times[-W,W]$. 
From~\eqref{pswfapprox} we obtain
\begin{equation}
\|\Proj f\|^2_{L^2(\R)}\geq 1-49\epsilon^2 \|f\|^2_{L^2(\R)}\label{boundeps}
\end{equation}
for all $f\in L^2(T,W,\epsilon)$. We write $\Proj_{\Phi}$ for the 
projection onto the set $\{\phi_1,\dots,\phi_N\}$ and $G$ for the
Gram matrix of $\{ \Proj \phi_1,...,\Proj \phi_N\}$, i.e. 
\begin{equation}
G_{i,j}=\langle \Proj \phi_j,\Proj \phi_i\rangle \;\;i,j=1,...,N.
\end{equation}
Then $\rank(\Proj \Proj_{\Phi})=\rank(G)$. 
Note that the diagonal entries of $G$ are positive and, since 
$\{\phi_1,...,\phi_N\}$ are orthonormal, that the eigenvalues of $G$ have absolute value at most $1$.   
By inequality (\ref{boundeps}) 
\begin{eqnarray*}
\sum_{j=1}^N G_{j,j}&=& \sum_{j=1}^N \|\Proj \phi_j\|^2\\
&\geq & N (1-49\epsilon^2),
\end{eqnarray*}
so that $\rank(G)\geq N(1-49\epsilon^2)$. 
Therefore, 
\begin{eqnarray*}
\rank(\Proj^{\perp}\Proj_{\Phi})&\leq & \rank(\Proj)-\rank (\Proj \Proj_\Phi)\\
&\leq & (2TW+1)-(1-49\epsilon^2)N, 
\end{eqnarray*}
and
\begin{eqnarray*}
\lefteqn{\| \Proj_\Phi \Proj\Ls\Ls^\ast \Proj\Proj_\Phi- \Proj\Ls\Ls^\ast \Proj\|_{\HS}}\\
&= & \| -\Proj_\Phi^{\perp} \Proj\Ls\Ls^\ast \Proj\Proj_\Phi+
\Proj\Ls\Ls^\ast\Proj\Proj_{\Phi}  -  \Proj\Ls\Ls^\ast\Proj\|_{\HS}\\
&=& \| -\Proj_\Phi^{\perp} \Proj\Ls\Ls^\ast\Proj\Proj_\Phi+\Proj\Ls\Ls^\ast\Proj\\
&&\hspace{3cm} -\Proj\Ls\Ls^\ast\Proj\Proj^{\perp}_{\Phi}  -
\Proj\Ls\Ls^\ast\Proj\|_{\HS}\\
&\leq& \| \Proj_\Phi^{\perp}
\Proj\Ls\Ls^\ast\Proj\Proj_\Phi\|_{\HS}+\|\Proj\Ls\Ls^\ast\Proj\Proj^{\perp}_{\Phi}  \|_{HS}\\
&\leq& 2\|\Proj_\Phi^{\perp} \Proj\|_{\HS}\|\Ls\|^2\\
&\leq & 2\;\rank(\Proj_\Phi^{\perp} \Proj)\|\Ls\|^2\\
&\leq & 2((2TW+1)-(1-49\epsilon^2)N)\|\Ls\|^2.
\end{eqnarray*}
If $2TW+1>N$, then set $\lambda_j(\Proj_\Phi \Ls\Ls^\ast \Proj_\Phi)=0$ for $N<j\leq 2TW+1$. 
Let $\pi$ be a permutation of the integers $1,...,2TW+1$. Then
\begin{small}
\begin{eqnarray*}
\lefteqn{\Big|\sum_{j=1}^{2TW+1} \log\big(1+\frac{\lambda_j( \Proj_\Phi \Ls\Ls^\ast\Proj_\Phi)}{\eta^2}\big)}\\
&&\hspace{2cm}-\sum_{j=1}^{2TW+1} \log\big(1+\frac{\lambda_{\pi(j)}( \Proj  \Ls\Ls^\ast \Proj)}{\eta^2}\big)\Big|\nonumber\\
&\leq& \sum_{j=1}^{2TW+1}\log\left(1+ \frac{| \lambda_j( \Proj_\Phi\Ls\Ls^\ast \Proj_\Phi)-\lambda_{\pi(j)}( \Proj  \Ls\Ls^\ast \Proj)|}{\eta^2}\right)\nonumber\\
&\leq& \sum_{j=1}^{2TW+1}\log\left(1+ \frac{| \lambda_j( \Proj_\Phi\Ls\Ls^\ast \Proj_\Phi)-\lambda_j( \Proj_\Phi \Proj \Ls\Ls^\ast \Proj
\Proj_\Phi)|}{\eta^2}\right.\\
&&\hspace{2cm}\left. +\frac{|\lambda_j( \Proj_\Phi \Proj\Ls\Ls^\ast
\Proj\Proj_\Phi)-\lambda_{\pi(j)}( \Proj  \Ls\Ls^\ast \Proj
)|}{\eta^2}\right).\notag 
 \end{eqnarray*}
\end{small}
We consider the first eigenvalue difference in the expression above.
Applying Theorem~A.46 in~\cite{BS10} we obtain
\begin{eqnarray}
\lefteqn{|\lambda_j( \Proj_\Phi \Proj\Ls\Ls^\ast
\Proj\Proj_\Phi)-\lambda_j( \Proj_\Phi \Proj\Ls\Ls^\ast\Proj) | }\nonumber\\
&&\hspace{1cm}\le  \|\Proj_\Phi \Ls \Ls^{\ast} \Proj_\Phi - \Proj_\Phi \Proj \Ls
\Ls^{\ast} \Proj \Proj_\Phi \|.\label{eigdiff1}
\end{eqnarray}
Let $\Pp f = u+v$ where $u \in \range \Proj$ and $v\in \range
\Proj^{\perp}$. Then
\begin{eqnarray}
\lefteqn{\|(\Pp \Ls \Ls^{\ast} \Pp - \Pp \Proj \Ls \Ls^{\ast} \Proj \Pp)f \|}\nonumber\\
& = & \|\Pp \Ls \Ls^{\ast}(u+v) - \Pp \Proj \Ls \Ls^{\ast} u \| \notag \\
&\le & \|(\Pp \Ls \Ls^{\ast} - \Pp \Proj \Ls \Ls^{\ast}) u \| +
 \|\Pp \Ls \Ls^{\ast} v \| \notag \\
& = & \|\Pp \Proj^\perp \Ls \Ls^{\ast} u \| + \|\Pp \Ls \Ls^{\ast} v \|
\notag \\
& \le & \|\Pp \Proj^\perp\| \| \Ls \|^2 \| u \| + \|\Pp \Ls \Ls^{\ast}
\Proj^\perp \Pp f \| \notag \\
& \le & 7 \epsilon \|\Ls\|^2 \|f\| + 7 \epsilon \|\Ls\|^2 \|f\| \notag \\
& \le & 14 \epsilon \|\Ls\|^2 \|f\|,\notag
\end{eqnarray}
where we have used~\eqref{pswfapprox} in the penultimate step. Hence
\begin{equation}
\label{eigdiff2}
|\lambda_j( \Proj_\Phi \Proj\Ls\Ls^\ast
\Proj\Proj_\Phi)-\lambda_j( \Proj_\Phi \Proj\Ls\Ls^\ast\Proj) | 
 \le  14 \epsilon \|\Ls\|^2.
\end{equation}

Concerning the second difference of eigenvalues
recall that according to Theorem~A.37 of~\cite{BS10} there exists 
a permutation $\pi$ such that 
\begin{eqnarray}
\lefteqn{\sum_{j=1}^{2TW+1} |\lambda_j( \Proj_\Phi \Proj\Ls\Ls^\ast
\Proj\Proj_\Phi)-\lambda_{\pi(j)}( \Proj  \Ls\Ls^\ast \Proj )|^2 } \hspace{1cm} \notag \\
&\le& \| \Proj_\Phi \Proj\Ls\Ls^\ast \Proj\Proj_\Phi- \Proj  \Ls\Ls^\ast\Proj \|^2_{\HS}
\notag \\
&\le&  4((2TW+1)-(1-49\epsilon^2)N)^2\|\Ls\|^4. \label{eigdiff3}
\end{eqnarray}
Using~\eqref{eigdiff2},~\eqref{eigdiff3} and the concavity of the $\log$ function we compute
{\small 
\begin{eqnarray*}
\lefteqn{\sum_{j=1}^{2TW+1}\log\Big( 1+ \frac{| \lambda_j( \Proj_\Phi\Ls\Ls^\ast \Proj_\Phi)-\lambda_j( \Proj_\Phi \Proj \Ls\Ls^\ast \Proj\Proj_\Phi)|}{\eta^2} }\\
&&\hspace{1cm} +\frac{|\lambda_j( \Proj_\Phi \Proj\Ls\Ls^\ast
\Proj\Proj_\Phi)-\lambda_{\pi(j)}( \Proj  \Ls\Ls^\ast \Proj
)|}{\eta^2}\Big) \notag \\
&\le&  \sum_{j=1}^{2TW+1}\log\Big( 1+ \frac{ 14\epsilon
\|\Ls\|^2}{\eta^2}\\
&&\hspace{1cm} +\frac{ |\lambda_j(\Proj_\Phi\Proj\Ls\Ls^\ast\Proj\Proj_\Phi)-\lambda_{\pi(j)}( \Proj_\Phi\Proj\Ls\Ls^\ast\Proj)|}{\eta^2} \Big) \notag \\
&\leq &  \sum_{j=1}^{2TW+1}\log\Big(1+ \frac{ 14\epsilon\|\Ls\|^2}{\eta^2}  \\
&&\hspace{2cm} +\frac{2 ((2TW+1)-(1-49\epsilon^2)N) \|\Ls\|^2}{\eta^2(2TW+1)}\Big)\label{estimate}
\end{eqnarray*}
}
We will return to \eqref{estimate} twice, taking $N$ to be the cardinality of $\Phi_{T,W}$ and of our constructed set. 


We look at the system $\{\psi^t_{k,l}\}$ from Section~\ref{ss:wh}. 
The signal $\psi^t_{k,l}$ is exponentially localized around the point 
$(\rho\ab l,\rho\ba k)$. 
We select those signals that are contained in $L^2(T,W,\epsilon)$. 
For some positive constants $\delta_1$ and $\delta_2$, these are those signals with indices 
$0\leq k\leq  \frac{T}{\rho}\ab-\delta_1$ 
and 
$0\leq |l|\leq \frac{W}{\rho}\ba-\delta_2$.
We set $K= \frac{T}{\rho}\ab-\delta_1$ 
and $L=\frac{W}{\rho}\ba-\delta_2$. 
We denote by $\Proj_{K,L}$ the projection operator from $L^2(R)$ onto the 
span of $\{\psi^t_{k,l}\}_{k=0,l=-L}^{K,L}$. 
Now we use~\eqref{estimate} twice: once with $N=(1-\delta)(2TW+1)$ for the cardinality of the set $\Phi_{T,W}$, 
as assumed in the statement of the theorem, 
and once for $\Psi_{K,L}$, where the cardinality satisfies 
\begin{eqnarray*}
K(2L+1)\geq  \frac{2TW+1}{\rho^2}-\delta_1\frac{2W}{\rho}\ba -\delta_2\frac{T}{\rho}\ab.
\end{eqnarray*} 
The arguments above then yield
\begin{small}\begin{eqnarray}
\lefteqn{\Big|\sum_{j=1}^{2TW+1} \log\Big(1+\frac{\lambda_j( \Proj_\Phi \Ls\Ls^\ast \Proj_\Phi)}{\eta^2}\Big)}\\
&&\hspace{1cm}-\sum_{j=1}^{2TW+1} \log\Big(1+\frac{\lambda_{\pi(j)}( \Proj_{K,L}  \Ls\Ls^\ast\Proj_{K,L} )}{\eta^2}\Big)\Big|\nonumber\\
&\leq& (2TW+1)\bigg(\log\Big(1+ \Big(\frac{ 14\epsilon }{\eta^2}+\frac{ 2(49\epsilon^2+\delta) }{\eta^2}\Big)\|\Ls\|^2\Big)\nonumber\\
&&\hspace{-.5cm}+\log\Big(1+ \Big(\frac{ 14\epsilon
}{\eta^2}+\frac{2(1-\frac{(1-49\epsilon^2)}{\rho^2}+\frac{\delta_1}{\rho}\ba +\frac{\delta_2}{\rho}\ab) }{\eta^2}\Big) \|\Ls\|^2  \Big)\bigg)\nonumber\\
\label{finalestimate}
\end{eqnarray}\end{small}
The estimation of the eigenvalues $\Proj_{K,L}\Ls\Ls^\ast \Proj_{K,L}$ 
is given by the Lemmas~\ref{one} and~\ref{difference}. Applying
these two lemmas together with inequality~(\ref{finalestimate}) 
complete the proof of the theorem. 
\end{proof}

\section*{Acknowledgement} \label{s:}

We would like to thank the anonymous referee as well as the editor, Helmut
B\"olcskei, for their very careful reading of the manuscript. Their
excellent
and constructive feedback significantly improved the contents and the
presentation of this paper.


\end{document}